\documentclass[a4paper,conference]{IEEEtran}
\ifCLASSINFOpdf
\else
\fi
\usepackage{amsmath}
\usepackage{graphicx,,epsfig,amsmath,latexsym,amssymb,verbatim,color}
\usepackage{dsfont}
\usepackage{theorem}\usepackage{url}
\newtheorem{definition}{Definition}
\newtheorem{proposition}[definition]{Proposition}
\newtheorem{lemma}[definition]{Lemma}

\newtheorem{theorem}[definition]{Theorem}
\newtheorem{corollary}[definition]{Corollary}

\def\squareforqed{\hbox{\rlap{$\sqcap$}$\sqcup$}}
\def\qed{\ifmmode\squareforqed\else{\unskip\nobreak\hfil
\penalty50\hskip1em\null\nobreak\hfil\squareforqed
\parfillskip=0pt\finalhyphendemerits=0\endgraf}\fi}
\def\endenv{\ifmmode\;\else{\unskip\nobreak\hfil
\penalty50\hskip1em\null\nobreak\hfil\;
\parfillskip=0pt\finalhyphendemerits=0\endgraf}\fi}
\newenvironment{proof}{\noindent \textbf{{Proof~} }}{\qed}
\newenvironment{remark}{\noindent \textbf{{Remark~}}}{\qed}
\newenvironment{example}{\noindent \textbf{{Example~}}}{\qed}

\mathchardef\ordinarycolon\mathcode`\:
\mathcode`\:=\string"8000
\def\vcentcolon{\mathrel{\mathop\ordinarycolon}}
\begingroup \catcode`\:=\active
  \lowercase{\endgroup
  \let :\vcentcolon
  }

\newcommand{\nc}{\newcommand}
\nc{\rnc}{\renewcommand}
\nc{\beg}{\begin{equation}}
\nc{\eeq}{{\end{equation}}}
\nc{\beqa}{\begin{eqnarray}}
\nc{\eeqa}{\end{eqnarray}}
\nc{\lbar}[1]{\overline{#1}}
\nc{\bra}[1]{\langle#1|}
\nc{\ket}[1]{|#1\rangle}
\nc{\ketbra}[2]{|#1\rangle\!\langle#2|}
\nc{\braket}[2]{\langle#1|#2\rangle}

\nc{\proj}[1]{| #1\rangle\!\langle #1 |}
\nc{\avg}[1]{\langle#1\rangle}
\nc{\Rank}{\operatorname{Rank}}
\nc{\smfrac}[2]{\mbox{$\frac{#1}{#2}$}}
\nc{\tr}{\operatorname{Tr}}
\nc{\ox}{\otimes}
\nc{\dg}{\dagger}
\nc{\dn}{\downarrow}
\nc{\cA}{{\cal A}}
\nc{\cB}{{\cal B}}
\nc{\cC}{{\cal C}}
\nc{\cD}{{\cal D}}
\nc{\cE}{{\cal E}}
\nc{\cF}{{\cal F}}
\nc{\cG}{{\cal G}}
\nc{\cH}{{\cal H}}
\nc{\cI}{{\cal I}}
\nc{\cJ}{{\cal J}}
\nc{\cK}{{\cal K}}
\nc{\cL}{{\cal L}}
\nc{\cM}{{\cal M}}
\nc{\cN}{{\cal N}}
\nc{\cO}{{\cal O}}
\nc{\cP}{{\cal P}}
\nc{\cQ}{{\cal Q}}
\nc{\cR}{{\cal R}}
\nc{\cS}{{\cal S}}
\nc{\cT}{{\cal T}}
\nc{\cX}{{\cal X}}
\nc{\cY}{{\cal Y}}
\nc{\cZ}{{\cal Z}}
\nc{\cW}{{\cal W}}
\nc{\csupp}{{\operatorname{csupp}}}
\nc{\qsupp}{{\operatorname{qsupp}}}
\nc{\var}{{\operatorname{var}}}
\nc{\rar}{\rightarrow}
\nc{\lrar}{\longrightarrow}
\nc{\polylog}{{\operatorname{polylog}}}
\nc{\1}{{\mathds{1}}}
\nc{\wt}{{\operatorname{wt}}}
\nc{\av}[1]{{\left\langle {#1} \right\rangle}}
\nc{\supp}{{\operatorname{supp}}}

\def\e{\epsilon}

\def\G{\Gamma}

\def\U{\Upsilon}

\def\O{\Omega}

\nc{\RR}{{{\mathbb R}}}
\nc{\CC}{{{\mathbb C}}}
\nc{\FF}{{{\mathbb F}}}
\nc{\NN}{{{\mathbb N}}}
\nc{\ZZ}{{{\mathbb Z}}}
\nc{\PP}{{{\mathbb P}}}
\nc{\QQ}{{{\mathbb Q}}}
\nc{\UU}{{{\mathbb U}}}
\nc{\EE}{{{\mathbb E}}}
\nc{\id}{{\operatorname{id}}}

\nc{\CHSH}{{\operatorname{CHSH}}}

\newcommand{\Op}{\operatorname}
\nc{\be}{\begin{equation}}
\nc{\ee}{{\end{equation}}}
\nc{\bea}{\begin{eqnarray}}
\nc{\eea}{\end{eqnarray}}
\nc{\<}{\langle}
\rnc{\>}{\rangle}
\nc{\Hom}[2]{\mbox{Hom}(\CC^{#1},\CC^{#2})}
\nc{\rU}{\mbox{U}}

\nc{\ob}[1]{#1}

\nc{\SEP}{{\text{SEP}}}
\nc{\NS}{{\text{NS}}}
\nc{\LOCC}{{\text{LOCC}}}
\nc{\PPT}{{\text{PPT}}}
\nc{\EXT}{{\text{EXT}}}
\nc{\Sym}{{\operatorname{Sym}}}

\nc{\ERLO}{{E_{\text{r,LO}}}}
\nc{\ERLOCC}{{E_{\text{r,LOCC}}}}
\nc{\ERPPT}{{E_{\text{r,PPT}}}}
\nc{\ERLOCCinfty}{{E^{\infty}_{\text{r,LOCC}}}}
\nc{\Aram}{{\operatorname{\sf A}}}

\begin{document}
\title{A semidefinite programming upper bound of quantum capacity}
\author{\IEEEauthorblockN{Xin Wang\IEEEauthorrefmark{1},
Runyao Duan\IEEEauthorrefmark{1}\IEEEauthorrefmark{2}}
\IEEEauthorblockA{\IEEEauthorrefmark{1}Centre for Quantum Computation and Intelligent Systems\\ Faculty of Engineering and Information Technology\\
University of Technology Sydney (UTS),
NSW 2007, Australia}
\IEEEauthorblockA{\IEEEauthorrefmark{2}UTS-AMSS Joint Research Laboratory for Quantum Computation and Quantum Information Processing\\ Academy of Mathematics and Systems Science\\ Chinese Academy of Sciences, Beijing 100190, China}
Email: xin.wang-8@student.uts.edu.au, runyao.duan@uts.edu.au}

\maketitle

\begin{abstract}
Recently the power of positive partial transpose preserving (PPTp) and no-signalling (NS) codes in quantum communication has been studied. We continue with this line of research and show that the NS/PPTp/NS$\cap$PPTp codes assisted zero-error quantum  capacity depends only on the non-commutative bipartite graph of the channel and the one-shot case can be computed efficiently by semidefinite programming (SDP). As an example, the activated PPTp codes assisted zero-error quantum  capacity is carefully studied. We then present a general SDP upper bound $Q_\Gamma$ of quantum capacity and show it is always smaller than or equal to the  ``Partial transposition bound'' introduced by Holevo and Werner, and the inequality could be strict. This upper bound is found to be additive, and thus is an upper bound of the potential PPTp assisted quantum capacity as well. We further demonstrate that $Q_\Gamma$ is strictly better than several previously known upper bounds for an explicit class of quantum channels. Finally, we show that $Q_\Gamma$ can be used to bound the super-activation of quantum capacity.
\end{abstract}

\IEEEpeerreviewmaketitle

\section{Introduction}
A fundamental problem in quantum information theory is to determine the quantum capacity of quantum channels. The quantum capacity of a noisy quantum channel is the highest rate at which it can convey quantum information reliably over asymptotically many uses of the channel. Quantum capacity is  complicated to evaluate since it is characterized by a  multi-letter, regularized expression, and it is not even known to be computable \cite{Cubitt2015}. Even for the low dimensional channels such as the qubit depolarizing channel, the quantum capacity remains unknown.

To deal with the intractable problem of determining quantum capacities of channels, assistance such as entanglement or classical communication have been introduced into the capacity problem \cite{Bennett2006,Leung2015c}. Particularly, positive partial transpose preserving (PPTp) and no-signalling (NS) codes assisted quantum capacity has been studied \cite{Leung2015c}, which regards a channel code as a bipartite operation with an encoder belonging to the sender and a decoder belonging to the receiver. 

Given an arbitrary quantum channel, the only known general computable upper bound is the partial transposition bound introduced in  \cite{Holevo2001}. Other known upper bounds \cite{Smith2008a,Sutter2014,Gao2015a,Bruß1998, Cerf2000, Wolf2007, Smith2008b,Tomamichel2015a}  all require specific settings to be tight and computable. For example, the upper bound from no cloning argument \cite{Bruß1998, Cerf2000} only behaves well at very high noise levels. Also, upper bound raised by approximate degradable quantum channels \cite{Sutter2014} can evaluate the quantum capacity of arbitrary channels based on the single-letter capacity and this usually works well just for approximate degradable quantum channels. Thus it is of great interest and significance to find an efficiently computable upper bound for quantum capacity.

Before we present our main results, let us first review some notations and preliminaries. Let $\cN(\rho)=\sum_k E_k\rho E_k^\dag$ be a quantum channel from $\cL(A')$ to $\cL(B)$, where $\sum_k E_k^\dag E_k=\1_{A'}$. The Choi-Jamio\l{}kowski matrix of $\cN$ is given by $J_{AB}=\sum_{ij} \ketbra{i}{j}_A \ox \cN(\ketbra{i}{j}_{A'})=(\text{id}_{A}\ox\cN)\proj{\Phi_{AA'}}$, where $A$ and $A'$ are isomorphic Hilbert spaces with respective orthonormal basis $\{\ket i\}$ and {$\{\ket j\}$}, and $\ket{\Phi_{AA'}}=\sum_k\ket{k_A}\ket{k_{A'}}$ is the unnormalized  maximally-entangled state over $A \ox A'$.
And $K=K(\cN)=\operatorname{span}\{E_k\}$ denotes the Choi-Kraus operator space of $\cN$. 
The coherent information of $\cN$ is given by 
\begin{equation}
\Op{I}_{\textup{\tiny C}}(\cN) 
= \max_{\rho_{A}} \Op{H}(\cN(\rho_{A})) - \Op{H}(\cN^c(\rho_{A}))  ,
\end{equation}
where $\cN^c$ is the complementary channel of $\cN$ and $\Op{H}(\sigma)=-\tr (\sigma \log \sigma)$ denotes
the von~Neumann entropy of a density operator $\sigma$. The work in \cite{Lloyd1997,Shor2002a,Devetak2005a} showed that coherent information of $\cN$ is an achievable rate for quantum communication while the work in \cite{Schumacher1996a,Barnum2000,Barnum1998} showed the  regularized coherent information is also an upper bound on quantum capacity. This establishes that
\begin{equation}
  \Op{Q}(\cN) = \lim_{n\rightarrow\infty} 
  \frac{\Op{I}_{\textup{\tiny C}}(\cN^{\otimes n})}{n} \,.
\label{eq:qcap}
\end{equation}

A general ``code''  is defined as a set of operations performed by the sender Alice and the receiver Bob which can be used to improve the data transmission with the given channel \cite{Leung2015c}. The PPTp codes are those for which the bipartite operation is PPT-preserving. 
A nonzero positive semi-definite operator $E \in \cL(\cX \ox \cY)$ is
said to be a positive
partial transpose operator (or simply PPT) if $E^{T_{\cX}}\geq 0$, where ${T_{\cX}}$ means the partial transpose with respect to the party
$\cX$, i.e., $(\ketbra{ij}{kl})^{T_{\cX}}=\ketbra{kj}{il}$.
A bipartite operation $\Pi:\cL(A_i\ox B_i)\rightarrow \cL(A_o\ox B_o)$ is `PPT-preserving' if it sends any state which is PPT with respect to the Alice/Bob partition to another PPT state. As shown in \cite{Rains2001},
a bipartite operation $\Pi^{A_i\ox B_i\rightarrow A_o\ox B_o}$ is PPT-preserving if and only if its Choi-Jamio\l{}kowski matrix $Z_{A_iB_iA_oB_o}$  is PPT.

The PPT-preserving operations include all operations
that can be implemented by local operations and classical communication (LOCC) and were introduced to study entanglement distillation in an early paper by Rains \cite{Rains2001}. They also include all unassisted and forward-classical-assisted codes introduced in \cite{Leung2015c}.  The no-signalling (NS) codes refer to the bipartite quantum operations with the no-signalling constraints and this kind of codes are also useful in classical zero-error communication \cite{Cubitt2011, Duan2016,Duan2015a}.
Let $\O$ represent $\text{NS, PPTp or NS}\cap\text{PPTp}$ in the rest of the paper.
Given a channel $\cN: \cL(A) \to \cL(B)$ and the $\O$ code of size $k$,  the optimal channel fidelity is given by the following SDP \cite{Leung2015c}:
\begin{equation}\label{PPT prime}
\begin{split}
F^{\O}(\cN, k)&= \max \tr J_{AB}W_{AB} \\ 
& \text{ s.t. }\  0 \leq W_{AB} \leq \rho_A \ox \1_B, \tr \rho_A=1,\\
&\text{PPTp: } -\frac{\rho_A \ox \1_B}{k} \le W_{AB}^{T_{B}} \le  \frac{\rho_A \ox \1_B}{k},\\
& \text{NS: }
\tr_A W_{AB} = \frac{1}{k^2}\1_B.
\end{split}\end{equation}
And the dual SDP is given by
\begin{equation}\begin{split}\label{ppt quantum dual}
F_d^{\O}(\cN, k)&= \min \mu +k^{-2}\tr S_B\\
\text{ s.t. }\  &J_{AB}+(Y_{AB}-V_{AB})^{T_{B}}\le X_{AB}+\1_A \ox S_B,\\
&  \tr_B (X_{AB}+k^{-1}(Y_{AB}+V_{AB}))\le \mu\1_A, \\
&  X_{AB}, Y_{AB}, V_{AB}\ge0.
\end{split}\end{equation}
To remove the \textbf{PPTp} constraint, set $Y_{AB}=V_{AB}=0$. To remove the \textbf{NS} constraint, set $S_B=0$. The strong duality holds for $F^{PPTp}(\cN,k)$, then $F^{PPTp}(\cN,k)=F_d^{PPTp}(\cN,k)$. 
Leung and Matthews \cite{Leung2015c} further introduced the quantum data transmission via quantum channels assisted with $\Omega$ codes.
The $\Omega$ codes assisted zero-error quantum  capacity is given by
   $$Q_0^{\O}(\cN)
        = \sup_n \max \left\{ \frac{1}{n} \log k_n
        : F^{\O}(\cN^{\ox n}, k_n) = 1 , k_n\ge0\right\}.$$
When $n=1$, $Q_0^{\O,(1)}(\cN)=\left\lfloor {\kappa^\O(\cN)} \right\rfloor $ is the one-shot $\O$ codes assisted zero-error quantum capacity, where
\begin{equation}\label{PPT K}
        \kappa^\O(\cN)
        := \max \left\{ k
        : F^{\O}(\cN, k) = 1, k\ge 0 \right\},
\end{equation}
and $\left\lfloor {\kappa^\O(\cN)} \right\rfloor$ means the integer part of $\kappa^\O(\cN)$.
The corresponding quantum capacity is given by
\begin{equation}
        Q^{\O}(\cN)
        := \sup \{r: \mathop {\lim }\limits_{n \to \infty }  
        : F^{\O}(\cN^{\ox n}, \left\lfloor2^{rn}\right\rfloor) = 1\}.
\end{equation}
    
The so-called ``non-commutative graph theory'' was first suggested in \cite{Duan2013a}. The non-commutative graph associated with the channel captures the zero-error communication properties, thus playing a similar role to confusability graph of a classical channel. The zero-error classical capacity of a quantum channel in the presence of quantum feedback only depends on the Choi-Kraus operator space of the channel \cite{Duan2015}.  That is to say, the Choi-Kraus operator space $K$ plays a role that is quite similar to the bipartite graph and  $K$ is alternatively called ``non-commutative bipartite graph'' \cite{Duan2016}.
Based on the idea in \cite{Yang2015}, we also define the potential $\O$ codes assisted quantum capacity
 \begin{equation}
 Q_p^{\O}(\cN) := \sup_{\cM} \left[Q_p^{\O}(\cN\ox\cM)-Q_p^{\O}(\cM) \right].
\end{equation}

In this paper,  we first connect $\O$ codes assisted zero-error quantum  capacity to the non-commutative bipartite graph. We then introduce the activated PPTp codes assisted zero-error quantum capacity. Furthermore, we present a general SDP upper bound $Q_\Gamma$  of quantum capacity.  A general upper bound is usually difficult to find, however, our upper bound $Q_\G$ can be applied to evaluate the quantum capacity of an arbitrary channel efficiently, whereas most previous upper bounds rely on specific conditions which can be different for each channel.  We show that $Q_\Gamma$ is always smaller than or equal to the  ``Partial transposition bound'' and the inequality can be strict. $Q_\G$ is additive under tensor product, and thus is an upper bound of the potential PPTp assisted quantum capacity. We also demonstrate that this SDP upper bound is strictly better than several known upper bounds by explicit examples. For the super-activation of quantum capacity \cite{Smith2008}, $Q_\Gamma$ can also be applied to evaluate the super-activation.
 

\section{Assisted zero-error quantum  capacity and non-commutative bipartite graph}
As non-commutative bipartite graphs play an important role in zero-error classical communication, we will investigate the relationship between  zero-error quantum capacity and non-commutative bipartite graph in this section. To be specific, we will prove that zero-error quantum capacities assisted with NS, PPTp or NS$\cap$PPTp codes also depend only on the non-commutative bipartite graph of a quantum channel.

Let $P_{AB}$ denote the projection onto the support of the Choi-Jamio\l{}kowski matrix of $\cN$, which means that $P_{AB}$ is completely determined by $K(\cN)$. We also define the following SDP which only depends on $K$,
\begin{equation}\label{O F prime}
\begin{split}
D^{\O}(K, k)&= \max \tr  P_{AB}(W_{AB}-\rho_A \ox \1_B)\\
&\text{ s.t. }\  0 \leq W_{AB} \leq \rho_A \ox \1_B, \tr \rho_A=1,\\
&\text{PPTp: } -\frac{\rho_A \ox \1_B}{k} \le W_{AB}^{T_{B}} \le  \frac{\rho_A \ox \1_B}{k},\\
& \text{NS: }
\tr_A W_{AB} = \frac{1}{k^2}\1_B.
\end{split}\end{equation}

\begin{theorem}\label{capacity by graph}
For a quantum channel $\cN$ with non-commutative bipartite graph $K$, 
$F^{\O}(\cN, k)=1$ if and only if
$D^{\O}(K, k) = 0$. 
Furthermore, 
$Q_0^{\O,{(1)}}(\cN)=Q_0^{\O,{(1)}}(K)
        = \left\lfloor {\kappa^{\O}(K)} \right\rfloor$, where $\kappa^{\O}(K)=\max \left\{ k: D^{\O}(K, k) = 0, k\ge 0 \right\}$.
\end{theorem}
\begin{proof}
Firstly, noting that
$\tr({\rho_A} \ox {\1_B}){J_{AB}} = \tr_A\tr_B[({\rho_A} \ox {\1_B}){J_{AB}}] =\tr{\rho_A}=1$,
we have that
\begin{align*}
F^{\O}(\cN, k)-1&= \max   \tr  J_{AB}(W_{AB}-\rho_A \ox \1_B)\\
&\text{ s.t. }\  0 \leq W_{AB} \leq \rho_A \ox \1_B, \tr \rho_A=1,\\
&\text{PPTp: } -\frac{\rho_A \ox \1_B}{k} \le W_{AB}^{T_{B}} \le  \frac{\rho_A \ox \1_B}{k},\\
&  \text{NS: }
\tr_A W_{AB} = \frac{1}{k^2}\1_B.
\end{align*}

It is evident that $F^{\O}(\cN, k)-1=0$ if and only if 
$\tr  J_{AB}(W_{AB}-\rho_A \ox \1_B)=0$. Noting that $W_{AB}-\rho_A \ox \1_B\le 0$, then $\tr  J_{AB}(W_{AB}-\rho_A \ox \1_B)=0$ is equivalent to 
$\tr P_{AB}(W_{AB}-\rho_A \ox \1_B)=0$.
Therefore, $F^{\O}(\cN, k)=1$ if and only if $D^{\O}(K, k) = 0$. 
Consequently, zero-error quantum capacity assisted with $\O$ codes also depends only on the non-commutative bipartite graph.
\end{proof}

\begin{theorem}\label{NS capacity}
The one-shot NS codes assisted quantum zero-error capacity of a non-commutative bipartite graph $K$ is given by the interger part of $\kappa^{NS}(K)=\sqrt{\U(K)}$, where $\U(K)$ is the NS assisted zero-error classical  capacity introduced in \cite{Duan2016}.
\end{theorem}
\begin{proof}
We can first simplify $\kappa^{NS}(K)$ to
\begin{align*}
\kappa^{NS}(K)&= \max k \ \text{ s.t. }\  0 \leq k^2W_{AB} \leq k^2\rho_A \ox \1_B,\\
&\phantom{=\max k \text{ s.t. }}  
\tr_A k^2W_{AB} = \1_B, \\
&\phantom{=\max k \text{ s.t. }}   \tr P_{AB}(k^2\rho_A \ox \1_B-k^2W_{AB}) = 0.
\end{align*}
Then suppose that $U_{AB}=k^2W_{AB}$ and $k^2\rho_A =S_A$, therefore
\begin{align*}
\kappa^{NS}(K)&= \max \sqrt {\tr S_A} \ \text{ s.t. }\  0 \leq U_{AB} \leq S_A\ox \1_B,\\
&\phantom{=\max k \ \text{ s.t. }}  
\tr_A U_{AB} = \1_B, \\
&\phantom{=\max k \ \text{ s.t. }}  \tr P_{AB}(S_A\ox \1_B-U_{AB}) = 0.
\end{align*}
Hence, $[\kappa^{NS}(K)]^2=\U(K)$.
\end{proof}

For a quantum channel $\cN$ assisted PPTp codes, we can ``borrow'' a noiseless qudit channel $I_d$ whose zero-error quantum capacity is $d$, then we can use $\cN\ox I_d$ to transmit information. After the communication finishes we ``pay back'' the capacity of $I_d$. This kind of communication method was suggested in \cite{Acin2015a,Duan2015a}, and was highly relevant to the notion of \textit{potential capacity} recently studied by Winter and Yang \cite{Yang2015}. Based on this model, we define
the one-shot activated PPTp codes assisted zero-error quantum capacity (message number form) is
\begin{equation}\label{activated K}
      \kappa_a^{PPTp}(\cN)
        := \mathop {\sup }\limits_{d \ge 2} \frac{\left\lfloor {\kappa^{PPTp} (\cN \ox {I_d})}\right\rfloor }{d}.
\end{equation}
where $I_d$ is a noiseless qudit channel.

\begin{proposition}\label{tensor qudit channel}
For a quantum channel $\cN$ and a qudit noiseless channel $I_d$,
$F^{PPTp}(\cN \ox I_d,kd)=F^{PPTp}(\cN,k)$.
Consequently,
$\kappa^{PPTp} (\cN \ox I_d)=d \kappa^{PPTp}(\cN)$.
\end{proposition}
\begin{proof}
On one hand, it is easy to prove that for two quantum channel $\cN_1$ and $\cN_2$, 
$$F^{PPTp}(\cN_1,k_1)F^{PPTp}(\cN_2,k_2)\le F^{PPTp}(\cN_1 \ox \cN_2 ,k_1k_2).$$
Thus, $F^{PPTp}(\cN,k)\le F^{PPTp}(\cN\ox I_d, kd)$.

On the other hand, suppose that $F^{PPTp}(\cN,k)=u$,
assume that the optimal solution to SDP (\ref{ppt quantum dual}) of $F^{PPTp}(\cN,k)$ is $\{X_1,Y_1,V_1\}$. 
For a Hermitian operator $Z$, we define the positive part $Z_+$ and the negative part $Z_-$ to be the unique positive operators such
that  $Z=Z_+-Z_-$ and $Z_+Z_-=0$.
Let  $X_2=0, Y_2=(\Phi_d^{T_{B'}})_{-}, V_2=(\Phi_d^{T_{B'}})_{+}$,
where $ \Phi_d$ is the unnormalized maximally entanglement $\proj{\Phi_d}$ with $\ket{\Phi_d}=\sum_{i=0}^{d-1} \ket{ii}$.
Then, $\{X_2,Y_2,V_2\}$ is a feasible solution to SDP (\ref{ppt quantum dual}) of $F^{PPTp}(I_d,d)$. 
Furthermore,
noting that $Y_2+V_2=(\Phi_d^{T_B'})_{-}+(\Phi_d^{T_B'})_{+}=\1_{BB'}$,
we can assume that $X=X_1\ox\Phi_d$,
$Y-V=-(Y_1-V_1)\ox(Y_2-V_2)=(Y_1-V_1)\ox\Phi_d^{T_{B'}}$ and
$Y+V=(Y_1+V_1)\ox(Y_2+V_2)=(Y_1+V_1)\ox\1_{BB'}$.
Then it is easy  to show  that $\{u,X,Y,V\}$ is a feasible solution to the dual SDP of $F^{PPTp}(\cN \ox I_d,kd)$.

Hence, $F^{PPTp}(\cN \ox I_d,kd)=u=F^{PPTp}(\cN,k)$.
\end{proof}

\begin{proposition}\label{ac capacity}
For a channel $\cN$, $\kappa_a^{PPTp}(\cN)= \kappa^{PPTp} (\cN)$.
Furthermore,
$Q_{0,a}^{PPTp}(\cN)=Q_{0}^{PPTp}(\cN)$. Then, 
$$Q_0^{PPTp}(\cN \ox I_d)=Q_0^{PPTp}(\cN) +\log d.$$
\end{proposition}
\begin{proof}
Let us first consider the case that $\kappa^{PPTp}(\cN)$ is a rational number. W.l.o.g, we assume that $\kappa^{PPTp}(\cN)=\frac{t}{m}$, where $t$ and $m$ are positive integers.
On one hand,
$$\kappa_a^{PPTp}(\cN)\ge {\left\lfloor\kappa^{PPTp}  (\cN) \kappa^{PPTp} (I_m)\right\rfloor}/{m}=\frac{t}{m}.$$
On the other hand, by Proposition \ref{tensor qudit channel}, we have
$$\kappa_a^{PPTp}(\cN) \le \mathop {\sup }\limits_{d \ge 1} [{{\kappa^{PPTp} (\cN \ox {I_d})}}/{d}]= \kappa^{PPTp} (\cN).$$
Hence, $\kappa_a^{PPTp}(\cN)=\kappa^{PPTp}  (\cN)$ and $Q_{0,a}^{PPTp}(\cN)=Q_{0}^{PPTp}(\cN)$.
Finally, the case of irrational numbers can be solved by taking limit and using continuity arguments.
\end{proof}

\begin{example}\label{WH capacity}
The $d$-dimensional Werner-Holevo channel is defined as
$\cW_{d}(\rho)=\frac{1}{d-1}(\1_B\tr \rho- \rho^T)$.
$\cW_{d}$ is anti-degradable and hence has no  quantum capacity. However, the asymptotic quantum capacity and the zero-error quantum capacity of PPT-preserving codes over $\cW_3$ are both $\log \frac{d+2}{d}$ \cite{Leung2015c}. For this $\cW_d$,
$$Q_{0}^{PPTp}(\cW_d)=\log\kappa_a^{PPTp}(\cW_d)=\log \frac{d+2}{d}.$$
We will first show a feasible solution $\{\rho_A, V_{AB}\}$ of $F^{PPTp}(\cW_d,\frac{d+2}{d})=1$. Let 
$\rho_{A}=\frac{1}{d}\1_{A}$ and  $V_{AB}=(\frac{1}{d+2}\1_{AB}-\frac{2}{d(d+2)} \Phi_d)^{T_B}$,
where $ \Phi_d$ is the unnormalized maximally entanglement $\proj{\Phi_d}$ with $\ket{\Phi_d}=\sum_{i=0}^{d-1} \ket{ii}$.
It is easy to check that $\{\rho_{A}, V_{AB}\}$ is a feasible solution such that $F^{PPTp}(\cW_d,\frac{d+2}{d})=1$, which means that
$\kappa^{PPTp}(\cW_d)\ge\frac{d+2}{d}$. Thus,
$\log\kappa_a^{PPTp}(\cW_d)=\log\frac{d+2}{d}=Q_{0}^{PPTp}(\cW_d)$.
\end{example}

\section{A general upper bound of quantum capacity}
Since computing the quantum capacity of a quantum channel is very difficult, we will introduce an SDP upper bound to evaluate the quantum capacity of any channel. Semidefinite programming (SDP) problems \cite{Vandenberghe1996}  can be solved by polynomial time algorithms \cite{Khachiyan1980}. The CVX software  \cite{Grant2008} allows one to solve SDPs efficiently.

To be specific, we define $Q_\G(\cN)=\log\G(\cN)$ and
\begin{equation}\begin{split}
\G(\cN)&= \max   \tr  J_{AB}R_{AB} \\ 
&\text{ s.t. }\   R_{AB},\rho_A\ge0, \tr{\rho_A}=1, \\
&\phantom{\text{ s.t. }} -\rho_A \ox \1_B \le R_{AB}^{T_{B}} \le \rho_A \ox \1_B.
\end{split}\end{equation}
The dual SDP is given by
\begin{equation}\label{dual WN}
\begin{split}
\G(\cN)&= \min   \mu \\
& \text{ s.t. }   Y_{AB},V_{AB}\ge0, (V_{AB}-Y_{AB})^{T_{B}}  \ge J_{AB},\\
&\phantom{\text{ s.t. }}\tr_B(V_{AB}+Y_{AB})\le \mu \1_A.
\end{split}\end{equation}
By strong duality, the values of both the primal and the dual SDP coincide. This quantity also relates to Rains bound \cite{Rains2001} and the improved SDP bound of distillable entanglement \cite{Wang2016}.
$Q_\G$ has some remarkable properties. For example, it is additive:
$Q_\G(\cN \ox \cM)=Q_\G(\cM)+Q_\G(\cN)$ for different quantum channels $\cN$ and $\cM$. This can be proved by utilizing semi-definite programming duality.

\begin{theorem}\label{ppt upper bound}
For quantum channels $\cM$ and $\cN$, 
$Q^{PPTp}(\cN)+Q^{PPTp}(\cM)\le Q^{PPTp}(\cM \ox \cN)\le Q^{PPTp}(\cM)+ Q_\G(\cN)$.

Consequently,
\begin{align*}
Q(\cN)&\le Q^{FCA}(\cN)\le Q^{FHA}(\cN)\\
&\le Q^{PPTp}(\cN)\le Q_p^{PPTp}(\cN)\le Q_\G(\cN),
\end{align*}
where \textbf{FCA}, \textbf{FHA} represent for forward-classical-assisted codes and forward-Horodecki-assisted codes, respectively.
\end{theorem}
\begin{proof}
Firstly, from SDP (\ref{PPT prime}), it is easy to see that $Q^{PPTp}(\cN)+Q^{PPTp}(\cM)\le Q^{PPTp}(\cM \ox \cN)$.

Secondly, assume that $Q^{PPTp}(\cM \ox \cN)=q$, then
$$ \mathop {\lim }\limits_{n \to \infty }  
        F^{PPTp}((\cN\ox\cM)^{\ox n}, \left\lfloor2^{qn}\right\rfloor) = 1.$$
Let $\G(\cN)=t$, from Lemma \ref{upper bound tensor F} below, we have that
\begin{align*}
    1\ge &\mathop {\lim }\limits_{n \to \infty }  
        F^{PPTp}(\cM^{\ox n}, \frac{\left\lfloor2^{qn}\right\rfloor}{t^n})\\
        \ge &  \mathop {\lim }\limits_{n \to \infty }  
        F^{PPTp}((\cN\ox\cM)^{\ox n}, \left\lfloor2^{qn}\right\rfloor)= 1. 
\end{align*}
Let $Q^{PPTp}(\cM)=r$, then from the definition,
\begin{equation}
\left\lfloor2^{rn}\right\rfloor\ge \frac{\left\lfloor2^{qn}\right\rfloor}{t^n}, n \to \infty.
\end{equation}
Then, it is easy to see that
$t2^{r}\ge (2^{qn}-1)^{1/n}$ ($n \to \infty$),
which means that  $\log t+r\ge q$.
Hence, $Q^{PPTp}(\cM \ox \cN) 
\le Q^{PPTp}(\cM)+ Q_\G(\cN)$.
Then we immediately have that
$Q^{PPTp}(\cN)\le  Q_p^{PPTp}(\cN)
=\sup_{\cM}[Q^{PPTp}(\cM \ox \cN)-Q^{PPTp}(\cM)] 
\le Q_\G(\cN)$.
\end{proof}

\begin{lemma}\label{upper bound tensor F}
For quantum channels $\cN_1$ and $\cN_2$, we have that
\begin{align*}
&F^{PPTp}(\cN_1,k)F^{PPTp}(\cN_2,\G(\cN_2))\\
\le  &F^{PPTp}(\cN_1\ox \cN_2, k\G(\cN_2)) 
\le  F^{PPTp}(\cN_1,k).
\end{align*}
\end{lemma}
\begin{proof}
It is easy to prove the first inequality. For the latter inequality, assume that the optimal solutions to dual SDPs of $F^{PPTp}(\cN_1,k)$  and $\G(\cN_2)$ are $\{u_1,X_1,Y_1,V_1 \}$ and $\{u_2,Y_2,V_2\}$, respectively. Let
$X=X_1 \ox J_2, V-Y=(V_1-Y_1)\ox (V_2-Y_2), Y+V=(Y_1+V_1)\ox (Y_2+V_2)$,
then the idea is to prove that $\{u_1, X, Y, V\}$ is a feasible solution to dual SDP of $F^{PPTp}(\cN_1\ox \cN_2, k\G(\cN))$, which means that $F^{PPTp}(\cN_1\ox \cN_2, k\G(\cN))\le F^{PPTp}(\cN_1,k)$.
\end{proof}

\begin{corollary}
For any two quantum channels $\cN$ and $\cM$, we have that
$Q^{PPTp}(\cN \ox \cM)\le Q_\G(\cN)+Q_\G(\cM)$. 
\end{corollary}
\begin{remark}
In \cite{Smith2008}, the super-activation of quantum capacity says that two zero-capacity channels (50\% erasure channel $\cN_{e}^{0.5}$ and a Horodecki channel $\cN_H$) can have a nonzero capacity when used together, i.e. $Q(\cN_{e}^{0.5} \ox \cN_H) >0.01$. Here, applying this corollary, we can evaluate the super-activation: $Q(\cN_{e}^{0.5} \ox \cN_H) \le Q_\G(\cN_{e}) +Q_\G(\cN_H)= Q_\G(\cN_{e}^{0.5})\approx 1.123$.
\end{remark}

\section{Comparison with other bounds}
In \cite{Holevo2001}, Holevo and Werner gave a general upper bound of quantum capacity for channel $\cN$ with Choi-Jamio\l{}kowski matrix $J_{\cN}$:
\begin{equation}
Q(\cN) \le Q_{\Theta}(\cN)=\log \| J_{\cN}^{T_B}\|_{cb}.
\end{equation}
Here $\|\cdot\|_{cb}$ is the completely bounded trace norm, which is known
to be efficiently computable by semidefinite programming \cite{Watrous2012}.

\begin{theorem}
For a quantum channel $\cN$, 
$$Q(\cN)\le Q_{\G}(\cN) \le Q_{\Theta}(\cN),$$
and both inequalities can be strict.
\end{theorem}
\begin{proof}
Assume that the optimal solution of $\G(\cN)$ is $\{R_{AB},\rho_A\}$, then 
$\G(\cN)=\tr J_{\cN} R_{AB}=\tr J_{\cN}^{T_B} R_{AB}^{T_B}$.

From Theorem 6 in \cite{Watrous2012}, 
\begin{equation}\begin{split}\label{cb norm}
 \| J_{\cN}^{T_B}\|_{cb}= \max \frac{1}{2}\tr (J_{\cN}^{T_B} X)+ \frac{1}{2}\tr (J_{\cN}^{T_B} X^\dagger)   \\
 \text{  s.t. } \left( {\begin{array}{*{20}{c}}
{\rho_0 \otimes \1 }&X\\
X^\dagger&{\rho_1 \otimes \1  }
\end{array}} \right)\ge 0.
\end{split}\end{equation}

Let us add two constraints $\rho_0=\rho_1=\rho_A$ and $X=X^\dagger$, then
$$
 \| J_{\cN}^{T_B}\|_{cb}\ge \max \tr (J_{\cN}^{T_B} X) \text{  s.t. } \left( {\begin{array}{*{20}{c}}
{\rho_A \otimes \1 }&X\\
X&{\rho_A \otimes \1  }
\end{array}} \right)\ge 0.$$

Noting that $-\rho_A \ox \1\le R_{AB}^{T_B}\le \rho_A \ox \1$, then
\begin{align*}
&\left( {\begin{array}{*{20}{c}}
{\rho_A \otimes \1}&R_{AB}^{T_B}\\
R_{AB}^{T_B}&{\rho_A \otimes \1}
\end{array}} \right)\\
=&\frac{1}{2}\left( {\begin{array}{*{20}{c}}
{\rho_A \otimes \1}+R_{AB}^{T_B}&{\rho_A \otimes \1}+R_{AB}^{T_B}\\
{\rho_A \otimes \1}+R_{AB}^{T_B}&{\rho_A \otimes \1}+R_{AB}^{T_B}
\end{array}} \right)\\
+&\frac{1}{2}\left( {\begin{array}{*{20}{c}}
{\rho_A \otimes \1}-R_{AB}^{T_B}&-({\rho_A \otimes \1}-R_{AB}^{T_B})\\
-({\rho_A \otimes \1}-R_{AB}^{T_B})&{\rho_A \otimes \1}-R_{AB}^{T_B}
\end{array}} \right)\ge 0.
\end{align*}

Therefore, $R_{AB}^{T_B}$ satisfies the constraint above, which means that
$ \| J_{\cN}^{T_B}\|_{cb}\ge \tr (J_{\cN}^{T_B} R_{AB}^{T_B})=\G(\cN)$.
We will further compare our semidefinite programming upper bound $Q_{\G}(\cN)$ to $Q_{\Theta}(\cN)$  in Fig. \ref{spec3}  based on $\cN_r=\sum_i E_i\cdot E_i^\dagger(0\le r\le 0.5)$ with $E_0=\proj 0+\sqrt r\proj 1$ and $E_1=\sqrt{1-r}\ketbra 0 1+\ketbra 1 2$.
\end{proof}

 \begin{figure}[ht]
 \center
 \includegraphics[width=5.8cm]{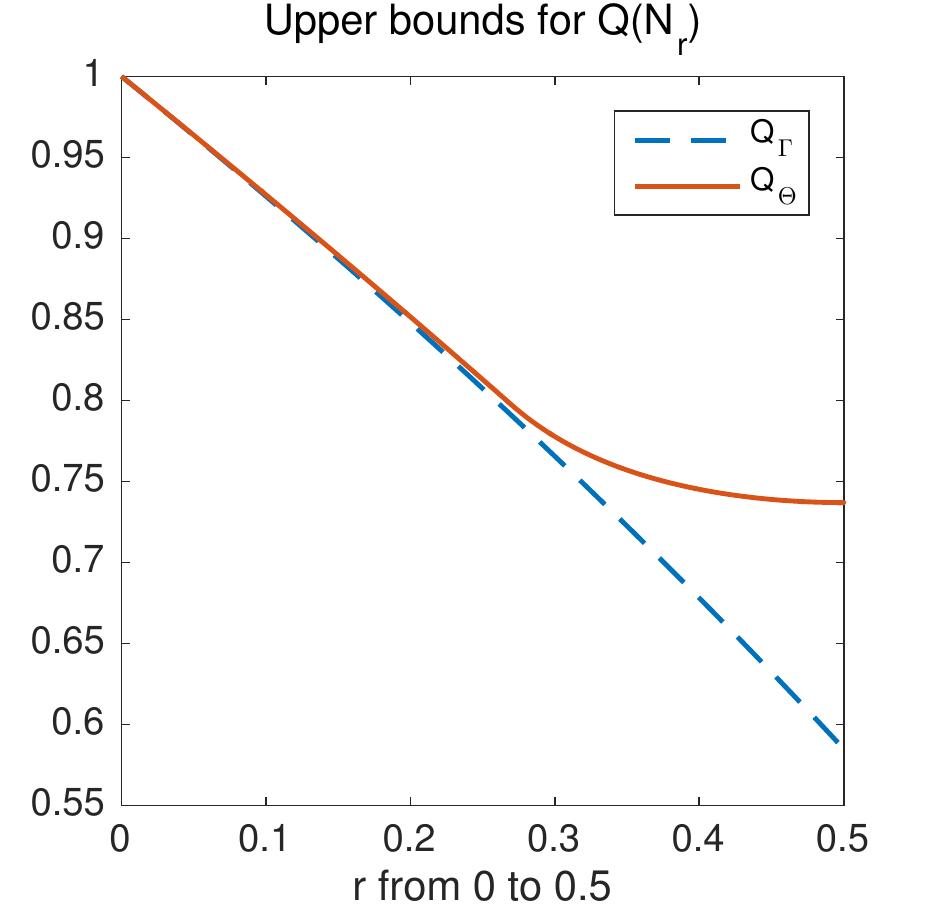}
 \caption{This plot shows different upper bounds of  $Q(\cN_r)$.  Dashed line depicts the upper bound $Q_\G(\cN_r)$ while solid line depicts $Q_\Theta(\cN_r)$}\label{spec3} 
\end{figure}

Comparing with the upper bound $Q_{AD}$ induced by $\e$-degradable quantum channels \cite{Sutter2014}, $Q_\G$ is tighter when $\e$ is not small. For example, for the class of channel $\cN_r$, when $r<0.38$, $Q_\G<\e\log 2+(1+\frac{1}{2}\e)h(\frac{\e}{2+\e})\le Q_{AD}$. 



\section{Conclusions}
We prove that the NS/PPTp/NS$\cap$PPTp codes assisted  zero-error quantum capacity depends only on the non-commutative bipartite graph of the channel and the NS codes assisted zero-error quantum capacity is given by the square root of the QSNC assisted zero-error classical capacity. We then introduce the activated PPTp codes assisted zero-error quantum capacity. Furthermore, we present a general SDP upper bound $Q_\Gamma$ of  quantum capacity, which can be used to evaluate the quantum capacity of an arbitrary channel efficiently.  $Q_\Gamma$ is always smaller than or equal to $Q_\Theta$ and can be strictly smaller than $Q_\Theta$ and $Q_{AD}$ for some channels. This upper bound is also additive and thus becomes an upper bound of the potential PPTp codes assisted capacity.  $Q_\Gamma$ can also be used to bound the super-activation of quantum capacity.

One interesting open problem is to determine the asymptotic PPTp codes assisted zero-error quantum  capacity $Q_{0}^{PPTp}(K)$.  Also, it would be very interesting to combine the upper bound $Q_\G$ with some entropy bounds such as the $Q_{ss}$ in \cite{Smith2008a}.

\section*{Acknowledgments}
We were grateful to M. Tomamichel for helpful suggestions.  This work was partly supported by the Australian Research Council (Grant No. DP120103776 and No. FT120100449) and the National Natural Science Foundation of China (Grant No. 61179030).



\bibliographystyle{IEEEtran}
 \bibliography{Q_upper_bound}

\end{document}